\documentclass[11pt,letterpaper]{article}
\usepackage[utf8]{inputenc}
\usepackage[english]{babel}
\usepackage{amsmath,amsthm,amsfonts,amssymb,dsfont,color}

\newcommand{\R}{{\ensuremath{\mathbb{R}}}}
\newcommand{\N}{{\ensuremath{\mathbb{N}}}}
\newcommand{\E}{{\ensuremath{\mathbb{E}}}}
\newcommand{\p}{{\ensuremath{\mathbb{P}}}}
\newcommand{\q}{{\ensuremath{\mathbb{Q}}}}
\newcommand{\qq}{\scriptscriptstyle {\ensuremath{\mathbb{Q}}}}

\newcommand{\calF}{\mathcal{F}}

\theoremstyle{definition}

\newtheorem{proposition}{Proposition}[section]
\newtheorem{corollary}{Corollary}[section]

\newtheorem{theorem}{Theorem}[section]

\newcommand{\rmd}{{\rm d}}
\newcommand{\rme}{{\rm e}}
\def\ep{\varepsilon}
\newcommand{\1}{\mathbf{1}}

\usepackage[noblocks]{authblk}

\begin{document}

\title{Jump-telegraph models for the short rate: pricing and convexity adjustments of zero coupon bonds}
\author[1]{Oscar L\'opez\thanks{ojlopeza@unal.edu.co}}
\author[2]{Gerardo E. Oleaga\thanks{goleaga@ucm.es}}
\author[1]{Alejandra S\'anchez\thanks{asanchezva@unal.edu.co}}

\affil[1]{\small Universidad Nacional de Colombia, Carrera 45 No. 26-85, Bogot\'a, Colombia}
\affil[2]{\small Departamento de An\'alisis Matem\'atico y Matem\'atica Aplicada \& Instituto de Matem\'atica Interdisciplinar, Universidad Complutense de Madrid, Spain}

\maketitle

\begin{abstract}
In this article, we consider a Markov-modulated model with jumps for short rate dynamics. We obtain closed formulas for the term structure and forward rates using the properties of the jump-telegraph process and the expectation hypothesis. The results are compared with the numerical solution of the corresponding partial differential equation.
\end{abstract}


\section{Introduction}\label{sec:Intro}

Regime changes associated to unexpected events in the economy provide an important field of study in quantitative finance. In the dynamics of the interest rates, the impact of those changes appears in the form of stochastic jumps. A well-known mathematical tool to model these types of events is provided by the family of Markov-modulated processes (see \cite{landen2000bond}, \cite{mamon2007hidden}). Among them, the so-called \emph{telegraphic processes with jumps} are of particular importance for a number of reasons (see \cite{kolesnik2013telegraph}, \cite{lopez2012kac}). First, they capture some of the stylized facts reported in the  literature (see \cite{wu2007exact}, \cite{johannes2004statistical} and the references therein); next, closed formulas for the primary descriptors (mean, variance, moment generating functions) are available; therefore, they are good candidates to model interest rate dynamics, providing explicit results for fixed income instruments.

In this work, we obtain formulas for the term structure modeled by the short rate dynamic under jump-telegraphic Markov-modulated processes. The models proposed (not necessarily affine) are such that the drift, jump size, and jump intensity depend on a continuous time Markov chain with two states. In our approach, jumps are due to a sudden change in the economic regime, that is, a switch in the market environment. For the applications considered, we require the integral of this process, whose primary characteristics are not known (see \cite{orsingher1999vector}). To bypass this difficulty, the so-called \emph{rational expectation hypothesis} is assumed. This implies that the forward instantaneous rate can be computed as the expected value of the future short rate. This approximation, also known as a \emph{convexity adjustment} (see \cite{pelsser2003mathematical}), is key to obtaining a closed formula for zero coupon bonds. This result is subsequently validated using a different approach. Under the no-arbitrage hypothesis, we provide the system of Partial Differential Equations (PDEs) that determine the price of the zero coupon bonds. A numerical approximation of the solutions to this system shows a close correspondence of both methods. In the second approach, the expectation hypothesis is not used; therefore, the convexity adjustment can be estimated.

This article is organized as follows. In the following section, the jump-telegraph models are defined; some of their properties are shown and some results connected to equivalent measures and martingales are proven.
In Section \ref{sec3} a general model for the short rate evolution is presented, involving a Markov-modulated process with jumps. The change of measure and an extension of the Feynman--Ka\v{c} theorem are analyzed. In Section \ref{sec4}, we provide analytic formulas for the forward rates in four particular cases of the general model. In Section \ref{sec5}, we provide formulas for the zero coupon bonds, as well as numerical approximations for the PDE system obtained for each example in Section \ref{sec4}. Conclusions are provided in the last section.


\section{Jump-telegraph process: characterization and properties}\label{sec:JT}
We present herein the jump-telegraph process and explore two types of technical properties. The first group (Propositions 2.1 and 2.2) is connected with its primary parameters and moment generating functions. These are required to provide an approximating formula for the forward rate in Section \ref{sec4}.  The second group of results (Theorem 2.1 and its corollary) will be used to provide a characterization of the equivalent probability measures in Section \ref{sec3}. This is performed through a change of measure process with specific martingale properties. 

Let $T>0$ be a fixed time and $(\Omega,\mathcal{F},\{\mathcal{F}_t\}_{t\in[0,T]},\p)$ a filtered probability space under the typical hypotheses. In this space, we consider a Markov chain in continuous time $\ep=\{\ep(t)\}_{t\in[0,T]}$, with state space $\{0,1\}$ and infinitesimal generator given by $\Lambda:=\left(\begin{smallmatrix}-\lambda_0 & \lambda_0 \\ \lambda_1 & -\lambda_1\end{smallmatrix}\right)$, $\lambda_0$, $\lambda_1>0$. 
Denoted by $\left\{ \tau_{n}\right\} _{n\geq1}$ the switching times of the Markov chain $\varepsilon$ and defining $\tau_0:=0$,
it is known that the inter-arrival times $\left\{ \tau_{n}-\tau_{n-1}\right\} _{n\geq1}$ are independent, exponential random variables, with
$\mathbb{P}\left\{ \tau_{n}-\tau_{n-1}>t\mid\varepsilon\left(\tau_{n-1}\right)=i\right\} =e^{-\lambda_{i}t}$, $i\in\{ 0,1\}$.

Let $N=\left\{ N_{t}\right\} _{t\in[0,T]}$ be the process that counts the number of state switches of the chain $\varepsilon$, defined as 
\begin{equation*}
N_{t}:=\sum_{n\geq1} \1_{\left\{ \tau_{n}\leq t\right\} }, \quad N_{0}=0.
\end{equation*}
$N$ is a non-homogeneous Poisson process with stochastic intensity given by $\{\lambda_{\varepsilon\left(t\right)}\}_{t\in[0,T]}$, a special case of the so-called \emph{doubly stochastic Poisson processes} see \cite{bremaud1981point}.

Let $c_0$, $c_1$, $h_0$, and $h_1$ be real numbers such that $c_0\not= c_1$, and $h_0,h_1\not=0$, and $\ep_j=\ep(\tau_j-)$ is the value of the Markov chain $\ep$ just before the $j$-th change at time $\tau_j$.  We define the jump-telegraph process $Y=\{Y_t\}_{t\in[0,T]}$ as the sum
\begin{equation*}
Y_t:=X_t+J_t=\int_0^t c_{\ep(s)}\,\rmd s+\sum_{j=1}^{N_t}h_{\ep_j},
\end{equation*}
where $X=\{X_t\}_{t\in[0,T]}$ is known as an \textit{asymmetric telegraph process} \cite{lopez2014asymmetric}, and  $J=\{J_t\}_{t\in[0,T]}$ is a pure jump process \cite{lopez2012kac}.

By fixing the initial state $\ep(0)=i\in\{0,1\}$, we have the following equality in the distribution:
\begin{equation}\label{eq:edY}
Y_t\overset{D}{=}c_it\1_{\{t<\tau_1\}}+\bigl[c_i\tau_1+h_i+\widetilde{Y}_{t-\tau_1}\bigr]\1_{\{t>\tau_1\}},
\end{equation}
for any $t>0$, where  $\widetilde{Y}=\{\widetilde{Y}_t\}_{t\in[0,T]}$ is a jump-telegraph process independent of $Y$, driven
by the same parameters, but starting from the opposite initial state $1-i$.

We denote by $p_i(x, t)$ the following density functions:
\begin{equation*}
p_i(x,t):=\frac{\p_i\{Y_t\in \rmd x\}}{\rmd x},\quad i=0,1,
\end{equation*}
where $\p_i\{\cdot\}=\p\{\cdot\mid \ep(0)=i\}$. That is, for any Borel set $\Delta$, $\Delta\subset\R$, $\int_{\Delta} p_i(x,t)\rmd x=\p_i\{Y_t\in\Delta\}$.

By \eqref{eq:edY} and the total probability theorem, the functions $p_i(x,t)$ satisfy the following system of integral equations on $[0,T]\times \R$
\begin{equation*}
p_i(x,t)=\rme^{-\lambda_it}\delta(x-c_it)+\int_0^t p_{1-i}(x-c_is-h_i,t-s)\lambda_i\rme^{-\lambda_is}\rmd s, \quad i=0,1,
\end{equation*}
where $\delta(\cdot)$ is Dirac's delta function.

We denote by $\E_i[\cdot]=\E[\cdot\mid \ep(0)=i]$, $i=0,1$ the conditional expectations under the initial value of the Markov chain $\ep$.
\begin{proposition}[\cite{kolesnik2013telegraph} Section 4.1.2]\label{pro:EY}
For any $t>0$, the conditional expectations $m_i(t):=\E_i[Y_t]$, $i=0,1$ of the jump-telegraph process $Y$ are given by
\begin{equation}\label{eq:solvmY}
m_i(t)=\frac{1}{2\lambda}\left[(\lambda_1d_0+\lambda_0d_1)t+(-1)^i\lambda_i(d_0-d_1)\left(\frac{1-\rme^{-2\lambda t}}{2\lambda}\right)\right],
\end{equation}
where $2\lambda:=\lambda_0+\lambda_1$ and $d_i=c_i+\lambda_ih_i$, $i=0,1$.
\end{proposition}

\begin{proposition}[\cite{lopez2012kac} Theorem 3.1]\label{pro:ME}
For any $z\in\mathbb{R}$ and $t>0,$ the moment generating
functions $\phi_{i}(z,t):=\E_{i}[\rme^{zY_{t}}]$, $i=0,1$ of the jump-telegraph process $Y$ are given by
\begin{equation}\label{solgenmom}
\phi_{i}\left(z,t\right)=\rme^{t\left(cz-\lambda\right)}\left(\cosh\left(t\sqrt{D}\right)+(-1)^i\left(az-\kappa+(-1)^i\lambda_{i}\rme^{zh_{i}}\right)\frac{\sinh\left(t\sqrt{D}\right)}{\sqrt{D}}\right),
\end{equation}
where $D=\left(az-\kappa\right)^{2}+\lambda_{0}\lambda_{1}\rme^{zH}$, $2c=c_0+c_1$, $2a=c_0-c_1$,
$2\kappa=\lambda_{0}-\lambda_{1}$ and $H=h_0+h_1$.
\end{proposition}
We define the filtration $\mathbf{F}=\{\mathcal{F}_t\}_{t\in[0,T]}$ generated together by the Markov chain and the Poisson process:
\begin{equation*}
\mathcal{F}_t=\sigma(\ep(s),s\in[0,t]) \vee \sigma(N_s, s\in[0,t])
\end{equation*} 

The following two results are key to obtaining a set of equivalent measures (see Proposition 1.7.1.1 in \cite{jeanblanc2009mathematical}). 

\begin{theorem}
The following processes are $\mathbf{F}$-martingales
\begin{align}
Z_t&:=\sum_{j=1}^{N_t}h_{\ep_j}-\int_0^t h_{\ep(s)}\lambda_{\ep(s)}\rmd s, \label{eq:Zcom}\\
\mathcal{
E}_t(Z)&:=\exp\left(-\int_0^t h_{\ep(s)}\lambda_{\ep(s)}\rmd s\right)\prod_{j=1}^{N_t}(1+h_{\ep_j}),\quad \text{for}\quad h_0,h_1>-1. 
\label{eq:expZcom}
\end{align}
Here, $\mathcal{E}_t(\cdot)$ denotes the stochastic (Dol\'eans-Dade) exponential (see, e.g., \cite{jeanblanc2009mathematical}, Section 9.4.3).
\end{theorem}
\begin{proof}
Observe that $Z$ is a jump-telegraph process with $c_i=-h_i\lambda_i$, $i=0,1$. Subsequently, by Proposition \ref{pro:EY}, $\E_i[Z_t]=0$, $i=0,1$. Let $s$, $0\leq s\leq t$ be fixed. Let $i\in\{0,1\}$ be the value of $\ep$ at time $s$ and let $k\in\N$ be the value of $N$ at time $s$. By the strong Markov property, we have the following conditional identities in the distribution   
\begin{equation}\label{eq:MSP}
\begin{split}
\bigl.\ep(s+u)\,\bigr|_{\{\ep(s)=i\}}\overset{D}{=}\bigl.\tilde{\ep}(u)\,\bigr|_{\{\tilde{\ep}(0)=i\}},\quad &\bigl.N_{s+u}\,\bigr|_{\{\ep(s)=i\}}\overset{D}{=}N_s+\bigl.\widetilde{N}_u\,\bigr|_{\{\tilde{\ep}(0)=i\}},\quad u\geq0,\\
\bigl.\tau_{k+j}\,\bigr|_{\{\ep(s)=i\}}&\overset{D}{=}\tilde{\tau}_{j}\,\bigr|_{\{\tilde{\ep}(0)=i\}}, \quad j=1,2,\dots
\end{split}
\end{equation}
where $\tilde{\ep}$, $\widetilde{N}$ and $\{\tilde{\tau_j}\}$ are copies of the processes $\ep$, $N$, and $\{\tau_j\}$, independents of $\mathcal{F}_s$. Subsequently, using the zero conditional property of $Z$ and \eqref{eq:MSP}, we obtain
\begin{equation*}
\E[Z_t-Z_s\mid\mathcal{F}_s]=\E_i\left[\sum_{j=1}^{\widetilde{N}_{t-s}} h_{\tilde{\ep}_j}-\int_0^{t-s} h_{\tilde{\ep}(u)}\lambda_{\tilde{\ep}(u)}\rmd u\right]=0,
\end{equation*}
and the first part follows. 

The jump-telegraph process is defined as follows:
\begin{equation*}
\widehat{Z}_t=\sum_{j=1}^{N_t}\log(1+h_{\ep_j})-\int_0^t h_{\ep(s)}\lambda_{\ep(s)}\rmd s,
\end{equation*}
Subsequently, we have $\mathcal{E}_t(Z)=\rme^{\widehat{Z}_t}$ and by Proposition \ref{pro:ME} we obtain $\E_i[\rme^{\widehat{Z}_t}]=1$, $i=0,1$. Using this and \eqref{eq:MSP}, we obtain
\begin{equation*}
\E\bigl[\rme^{\widehat{Z}_t-\widehat{Z}_s}\,\bigl|\bigr.\,\mathcal{F}_s\bigr]=\E_i\left[\exp\left(\,\sum_{j=1}^{\widetilde{N}_{t-s}}\log(1+h_{\tilde{\ep}_j})-\int_0^{t-s} h_{\tilde{\ep}(u)}\lambda_{\tilde{\ep}(u)}\rmd u\right)\right]= 1,
\end{equation*}
and the desired result is obtained.
\end{proof}

\begin{corollary}\label{LyEmarting}
The following processes are $\mathbf{F}$-martingales
\begin{align}
M_t&:=N_t-\int_0^t \lambda_{\ep(s)}\rmd s,\label{M}\\
L_t^\theta&:=\exp\left(\int_0^t(1-\theta_{\ep(s)})\lambda_{\ep(s)}\rmd s\right)\prod_{j=1}^{N_t}\theta_{\ep_j},\quad \text{for}\quad \theta_0,\theta_1>0.\label{Ltheta}
\end{align}
\end{corollary}
\begin{proof}
For the first part, notice that if we use $h_0=h_1=1$, by \eqref{eq:Zcom} the result is obtained. 

For the second part, let
\begin{equation*}
\widetilde{Z}_t=\sum_{j=1}^{N_t}(\theta_{\ep_j}-1)-\int_0^t(\theta_{\ep(s)}-1)\lambda_{\ep(s)}\rmd s.
\end{equation*}
Therefore, $L_t^\theta=\mathcal{E}_t(\widetilde{Z})$; subsequently, by \eqref{eq:Zcom} and \eqref{eq:expZcom}, the process $L^\theta$ is a $\p$-martingale with $\E_i[L_t^\theta]=1$, $i=0,1$.
\end{proof}


\section{General Markov-modulated jump-diffusion model for the short rate}\label{sec3}

\subsection{Short rate model}
For each $i\in\{0,1\}$, let $\mu_i:\Omega\times\R\to\R$, $\sigma_i:\Omega\times\R\to\R_{\geq0}$ and $\eta_i:\Omega\times\R\to\R\setminus\{0\}$ are measurable functions; we model the dynamics of the short interest rate by 
\begin{equation}\label{eq:dr}
\rmd r_t=\mu_{\ep(t)}(r_t)\rmd t+\sigma_{\ep(t)}(r_t)\rmd W_t+\eta_{\ep(t^{-})}(r_{t^{-}})\rmd N_t,
\end{equation}
where $\mu$ denotes the mean return rate, $\sigma$ the volatility, and $\eta$ the amplitude of the jumps of the short rate whenever a switch occurs in the Markov chain $\varepsilon$. Here, $W=\{W_t\}_{t\in[0,T]}$ is a standard Wiener process, independent of the Markov chain $\ep$ and hence of the Poisson process $N$.

We define the filtration $\mathbf{F}=\left\{ \mathcal{F}_{t}\right\}_{t\in[0,T]}$ now generated by the Markov chain, Poisson process, and Wiener process:
\begin{equation*}
\mathcal{F}_t=\sigma(\ep(s),s\in[0,t])\vee\sigma(W_s, s\in[0,t])\vee\sigma(N_s, s\in[0,t]).
\end{equation*}
The conditions that guarantee a unique strong solution of \eqref{eq:dr} for each initial value, can be found for instance in \cite{xi2011jump}, Proposition 2.1. Finally, the dynamics of the bank account is given by
\begin{equation*}
\rmd B_{t}=r_{t}B_{t}\rmd t.
\end{equation*}
\subsection{Change of measure and Feynman--Ka\v{c}'s theorem}
Let $M=\{M_t\}_{t\in[0,T]}$ be the martingale associated to the Poisson process $N$ (see Corollary \ref {LyEmarting}). Let us define a set of equivalent measures using the product of two Girsanov transforms. First, we define the change of measure process for the Wiener process $W$. Hence, we let $L^\psi=\{L^\psi_t\}_{t\in[0,T]}$ be the process defined by
\begin{equation}\label{eq:Lpsi}
L_t^\psi:=\exp\left(-\int_0^t\psi_{\ep(s)}\rmd W_s-\frac12\int_0^t\psi^2_{\ep(s)}\rmd s\right).
\end{equation}
where $\psi_0,\psi_1\in\R$. It is obvious that   $\{\psi_{\ep(t)}\}_{t\in[0,T]}$ satisfies Novikov's condition; therefore,  $L^\psi$ is a $\mathbf{F}$-martingale with $\E[L_T^\psi]=1$.

We now consider the martingale $L^\theta=\{L^\theta_t\}_{t\in[0,T]}$ defined by \eqref{Ltheta}. Subsequently, the process $L=\{L_t\}_{t\in[0,T]}$ defined by $L_t:=L_t^\psi\cdot L_t^\theta$ is an $\mathbf{F}$-martingale with $\E[L_T]=1$ (see \cite{jeanblanc2009mathematical} Section 10.3.1). Therefore, we define a set of equivalent measures using the Radon--Nikodym derivative $\left.\frac{\rmd \q}{\rmd \p}\right|_{\calF_t}:=L_t$, $t\in[0,T]$.

\begin{proposition}
Under measure $\q$, the process
\begin{equation}\label{eq:Wq}
W_t^{\qq}:=W_t-\int_0^t \psi_{\ep(s)} \rmd s, \quad t\in[0,T],
\end{equation}
is a Wiener process, and
\begin{equation}\label{eq:Mq}
M_t^{\qq}:=M_t-\int_0^t\bigl(1-\theta_{\ep(s)}\bigr)\lambda_{\ep(s)}\rmd s=N_t-\int_0^t \theta_{\ep(s)}\lambda_{\ep(s)}\rmd s,
\end{equation}
is a martingale.
\end{proposition}
\begin{proof}
For the first part, considering \eqref{eq:Lpsi}, applying the Girsanov theorem to the Wiener process, and using the independence of $W$ and $N$, we immediately obtain that \eqref{eq:Wq} is a Wiener process under $\q$. 

For the second part, using integration by parts and the independence of $W$ and $N$, we obtain
\begin{align*}
\rmd(M_t^{\qq}L_t^\theta)&=M_{t^-}^{\qq}\rmd L_t^\theta+L_{t^-}^\theta\rmd M_t^{\qq}+\rmd[M^{\qq},L^\theta]_t\\&=M_{t^-}^{\qq}\rmd L_t^\theta+L_{t^-}^\theta\rmd M_t^{\qq}+L_{t^-}^\theta(\theta_{\ep(t^-)}-1)\rmd N_t\\
&=M_{t^-}^{\qq}L_{t^-}^\theta(\theta_{\ep(t^-)}-1)\rmd M_t+L_{t^-}\rmd M_t+L_{t^-}^\theta(\theta_{\ep(t^-)}-1)\rmd M_t\\&=\bigl(M_{t^-}^{\qq}(\theta_{\ep(t^-)}-1)+\theta_{\ep(t^-)}\bigr)L_{t^-}^\theta\rmd M_t. 
\end{align*}
Thus, the process $M^{\qq}L^\theta$ is a $\p$-martingale and the process $M^{\qq}$ is a $\q$-martingale.
\end{proof}
Notice that, because of \eqref{eq:Mq}, the process $N$ under measure $\q$ is a Poisson process with intensities $\lambda^{\qq}_0=\theta_0\lambda_0$ and $\lambda^{\qq}_1=\theta_1\lambda_1$. Substituting in \eqref{eq:dr}, the dynamics of the short rate in the measure $\q$ is given by
\begin{equation}\label{eq:rQ}
\rmd r_t=\bigl(\mu_{\ep(t)}(r_t)+\sigma_{\ep(t)}(r_t)\psi_{\ep(t)}+\eta_{\ep(t)}(r_t)\theta_{\ep(t)}\lambda_{\ep(t)}\bigr)\rmd t+\sigma_{\ep(t)}(r_t)\rmd W_t^{\qq}+\eta_{\ep(t^{-})}(r_{t^{-}})\rmd M_t^{\qq}
\end{equation}

The particular structure of the short rate model proposed in \eqref{eq:dr} requires a special form of the Feynman--Ka\v{c} formula. To the best of our knowledge, such a formula is not available for this type of process in the current literature.

\begin{theorem}
We denote by $F(t,\ep(t),r_t):=F_{\ep(t)}(t,r_t)$ the price of a zero coupon bond with maturity $T$. Consider the coupled Cauchy problem: 
\begin{equation} \label{eq:DFi}
\begin{cases}
\dfrac{\partial F_{i}}{\partial t}(t,x)+\mathcal{L}F_i(t,x)=xF_i(t,x),\quad (t,x)\in[0,T)\times\R, \; i=0,1,\\[2mm]
F_{0}\left(T,x\right)=F_{1}\left(T,x\right)=1,
\end{cases}
\end{equation}
where $\mathcal{L}$ is the operator defined by
\begin{equation*}
\mathcal{L}F_i(t,x):=\bigl(\mu_{i}(x)+\psi_{i}\sigma_{i}(x)\bigr)\dfrac{\partial F_{i}}{\partial x}(t,x)+\dfrac{1}{2}\sigma_{i}^{2}(x)\dfrac{\partial^{2}F_{i}}{\partial x^{2}}\left(t,x\right) 
 +\lambda_{i}^{\qq}\bigl[F_{1-i}\left(t,x+\eta_{i}(x)\right)-F_{i}\left(t,x\right)\bigr].
\end{equation*}

If $F_i(\cdot,\cdot)$, $i=0,1$ is a solution to \eqref{eq:DFi}, subsequently the price of a zero coupon bond is represented as follows:
\begin{equation}\label{eq:FE}
F_{\ep(t)}\left(t,r_t\right)=\E^{\qq}\Bigl[\rme^{-\int_{t}^{T}r_s\rmd s}\mid\mathcal{F}_{t}\Bigr],
\end{equation}   
where the short rate process $r$ satisfies the stochastic differential equation \eqref{eq:rQ}.
\end{theorem}
\begin{proof}
By adapting the It\^o formula to the process $F_{\ep(t)}(t,r_t)$ (see for instance \cite{jeanblanc2009mathematical}, Section 10.2.2), we obtain
\begin{multline*}
F_{\ep(T)}(T,r_{T})  = F_{\ep(t)}(t,r_t) + \int_{t}^{T}\sigma_{\ep(s)}(r_s)\frac{\partial F_{\ep(s)}}{\partial x}(s,r_s)\rmd W_{s}^{\qq} \\ + \int_{t}^{T} \bigl[F_{1-\ep(s)} \left(s,r_{s^-} + \eta_{\ep(s)}(r_{s^-}) \right) -  F_{\ep(s)}\left(s,r_{s^-}\right)\bigr]\rmd M_{s}^{\qq} + \int_{t}^{T} \left[\frac{\partial F_{\ep(s)}}{\partial t}(s,r_s) + \mathcal{L}F_{\ep(s)}(s,r_s) \right] \rmd s  
\end{multline*}
Using \eqref{eq:DFi}, the last integral is equal to $\int_{t}^{T}r_sF_{\ep(s)}(s,r_s)\rmd s$. Consider now the process $Z=\{Z_t\}_{t\in[0,T]}$, defined as $Z_{t}:=\rme^{-\int_{0}^{t}r_s\rmd s}F_{\ep(t)}(t,r_t)$. By applying the product rule we obtain:
\begin{align*}
Z_{T} =\rme^{-\int_{0}^{t}r_s\rmd s}F_{\ep(s)}(t,r_t)&+\int_{t}^{T}\rme^{-\int_{0}^{s}r_u\rmd u}\sigma_{\ep(s)}(r_s)\frac{\partial F_{\ep(s)}}{\partial x}(s,r_s)\rmd W_{s}^{\qq}\\
&+\int_{t}^{T}\rme^{-\int_{0}^{s}r_u\rmd u}\bigl[F_{1-\ep(s)}\left(s,r_{s^-}+\eta_{\ep(s)}(r_{s^-})\right)-F_{\ep(s)}\left(s,r_{s^-}\right)\bigr]\rmd M_{s}^{\qq}.
\end{align*}
With the appropriate integrability conditions, the process $Z$ is a martingale and we finally obtain \eqref{eq:FE}.
\end{proof}
Notice that \eqref{eq:FE} provides the price of a zero coupon bond when the equivalent measure $\q$ is selected as a risk neutral measure, once the free parameters $\theta_i$ and $\Psi_i$ for $i=0,1$ are provided.

\section{Unbiased expectation hypothesis for jump-telegraph models}\label{sec4}

The unbiased expectation hypothesis postulates that, in an efficient market, the instantaneous forward rate at time $t$ with maturity $T$ must be equal to the expected future spot rate, that is  $f_{\ep(t)}(t,T,r_t)=\E^{\qq}[\,r_{\scriptscriptstyle T}\,|\,\mathcal{F}_t].$ Because the short rate is not deterministic, these values do not coincide, and their difference is a measure known as the convexity adjustment, see \cite{bjork2009arbitrage}, Section 26.4 and \cite{gaspar2008convexity}. 
We exploit this hypothesis and the jump-telegraph model to obtain a suitable analytical approximation of the forward rate.

\subsection{Jump-telegraph Merton model}
By adapting the classical Merton model for the short rate, our first model is established by the following stochastic differential equation:
\begin{equation*}
\rmd r_t=\mu_{\ep(t)}\rmd t+\eta_{\ep(t^{-})}\rmd N_t\,.
\end{equation*}
Notice that in this case, we have $\sigma_i=0$ for $i=0,1$; therefore, under the equivalent measure  $\q$, the dynamics of the short rate are maintained; however, the intensities of the process $N$ are $\lambda_i^\q = \theta_i\lambda_i$. 
The solution is given by
\begin{equation*}
r_t=r_0+X_t+J_t,
\end{equation*}
where 
\begin{equation*}
X_t=\int_0^t \mu_{\ep(s)}\rmd s\quad\text{y}\quad J_t=\int_0^t \eta_{\ep(s^-)}\rmd N_s=\sum_{j=1}^{N_t}\eta_{\ep_j}.
\end{equation*}
Subsequently, for all $t$, $t\in[0,T]$, we have that
\begin{equation}\label{eq:Er}
\E^{\qq}\left[r_{T}\mid\mathcal{F}_{t}\right] =\E^{\qq}\left[r_{t}+\int_{t}^{T}\mu_{\varepsilon\left(s\right)}\rmd s+\int_{t}^{T}\eta_{\varepsilon\left(s^{-}\right)}\rmd N_{s}\,\Bigl|\,\mathcal{F}_{t}\Bigr.\right]
\end{equation}
By the Markov property and applying the distributional equalities in \eqref {eq:MSP}, we can write \eqref{eq:Er} as follows:
\begin{align}
\E^{\qq}\left[r_{T}\mid\mathcal{F}_{t}\right] &=r_{t}+\E_i^{\qq}\Biggl[\int_{0}^{T-t}\mu_{\tilde{\varepsilon}(s)}\rmd s+\sum_{j=1}^{\widetilde{N}_{T-t}}\eta_{\tilde{\varepsilon}_{j}}\Biggr] \notag\\
&=r_{t}+\E_{i}^{\qq}\Bigl[\widetilde{X}_{T-t}+\widetilde{J}_{T-t}\Bigr]=r_{t}+\E_{i}^{\qq}\Bigl[\widetilde{Y}_{T-t}\Bigr]. \label{eq:HEM1}
\end{align}  
By Proposition \ref{pro:EY}, we have that the conditional expectation is given by 
\begin{align}
\E_i^{\qq}&\Bigl[\widetilde{Y}_{T-t}\Bigr]=\notag \\
&\frac{1}{2\lambda^{\qq}}\left[(\lambda_1^{\qq}d_0+\lambda_0^{\qq}d_1)(T-t)+(-1)^i\lambda_i^{\qq}(d_0-d_1)\biggl(\frac{1-\exp(-2\lambda^{\qq} (T-t))}{2\lambda^{\qq}}\biggr)\right], \label{eq:mYtilde}
\end{align}
for $i=0,1$, where $2\lambda^{\qq}=\lambda_0^{\qq}+\lambda_1^{\qq}$ y $d_i=\mu_i+\lambda_i^{\qq}\eta_i$, $i=0,1$.

\subsection{Jump-telegraph Dothan model}\label{ejG}
By adapting the classical Dothan model for the short rate, our second model is given by the following stochastic differential equation that we term the \emph{jump-telegraph Dothan's model}:
\begin{equation*}
\rmd r_t=r_{t^-}\bigl(\mu_{\ep(t)}\rmd t+\eta_{\ep(t^{-})}\rmd N_t\bigr).
\end{equation*}
As in the previous example, the dynamic under $\q$ is conserved, and only the intensities of the process $N$ are affected.
It is possible to write the solution in terms of the stochastic exponential as follows:
\begin{align*}
r_{t}=r_{0}\mathcal{E}_{t}\left(X+J\right)&=r_{0}\exp\left(\int_{0}^{t}\mu_{\varepsilon\left(s\right)}\rmd s\right)\prod_{j=1}^{N_{t}}\left(1+\eta_{\varepsilon_{j}}\right) \\
&=r_{0}\exp\left(\int_{0}^{t}\mu_{\varepsilon\left(s\right)}\rmd s+\int_{0}^{t}\log\left(1+\eta_{\varepsilon\left(s^{-}\right)}\right)\rmd N_{s}\right). 
\end{align*}
Let $\widehat{J}_{t}=\sum_{j=1}^{N}\log\left(1+\eta_{\varepsilon_{j}}\right)$;  subsequently,
$ r_{t}=r_{0}\exp\left(X_{t}+\widehat{J}_{t}\right)$. So that, for each $t\in[0,T]$ we have that
\[
r_{T}=r_{t}\exp\left(\int_{t}^{T}\mu_{\varepsilon\left(s\right)}ds+\int_{t}^{T}\log\left(1+\eta_{\varepsilon\left(s^{-}\right)}\right)dN_{s}\right)
\]
Subsequently, the expected value is as follows:
\begin{align*}
\E^{\qq}\left[r_{T}\mid\mathcal{F}_{t}\right] =r_{t}\E^{\qq}\left[\exp\left(\int_{t}^{T}\mu_{\varepsilon\left(s\right)}ds+\int_{t}^{T}\log\left(1+\eta_{\varepsilon\left(s^{-}\right)}\right)\right)dN_{s}\mid\mathcal{F}_{t}\right]
\end{align*}
By the Markov property and \eqref{eq:MSP}, we can write
\begin{align}
\E^{\qq}\left[r_{T}\mid\mathcal{F}_{t}\right] & =r_{t}\E_i^{\qq}\left[\exp\left(\int_{0}^{T-t}\mu_{\widetilde{\varepsilon}\left(s\right)}ds+\sum_{j=1}^{\widetilde{N}_{T-t}}\log\left(1+\eta_{\widetilde{\varepsilon}}\right)\right)\right] \notag\\
 & =r_{t}\E_i^{\qq}\left[\exp\left(\widetilde{X}_{T-t}+\widetilde{J}_{T-t}\right)\right]=r_t\E_{i}^{\qq}\Bigl[\exp\left(\widetilde{Y}_{T-t}\right)\Bigr]. \label{eq:HEM2}
\end{align}
By Proposition \ref{pro:ME}, we find that the exponential 
\begin{multline}
\E_{i}^{\qq}\Bigl[\exp\left(\widetilde{Y}_{T-t}\right)\Bigr] 
=\exp\left( (T-t)(\zeta-\lambda^{\qq})\right)\left[\cosh\left((T-t)\sqrt{D}\right) \right. \\ \left. +(-1)^i\left(\chi-\nu  + (-1)^i\lambda_{i}^{\qq}\left(1+\eta_{i}\right)\right)\frac{\sinh\left((T-t)\sqrt{D}\right)}{\sqrt{D}}\right] \label{eq:ME}
\end{multline}
for $i=0,1$, where $2\zeta=\mu_{0}+\mu_{1}$, $2\chi=\mu_{0}-\mu_{1}$,
$2\nu=\lambda_{0}^{\qq}-\lambda_{1}^{\qq}$, 
$D=\left(\chi-\nu\right)^{2}+\lambda_{0}^{\qq}\lambda_{1}^{\qq}\left(1+\eta_{0}\right)\left(1+\eta_{1}\right)$.

\subsection{Jump-telegraph diffusion Merton model}
Now, we add a diffusion term to the jump-telegraph Merton model:
\begin{equation*}
\rmd r_t=\mu_{\ep(t)}\rmd t+\sigma_{\ep(t)}\rmd W_t+\eta_{\ep(t^{-})}\rmd N_t\,.
\end{equation*}
Under measure $\q$, this is given by
\begin{equation*}
\rmd r_t=(\mu_{\ep(t)}+\sigma_{\ep(t)}\psi_{\ep(t)})\rmd t+\sigma_{\ep(t)}\rmd W_t^{\qq}+\eta_{\ep(t^{-})}\rmd N_t
\end{equation*}
The solution is now written as follows:
\begin{equation*}
r_{t}=r_{0}+X_t+Z_t+J_t,
\end{equation*}
where 
\begin{equation*}
X_t=\int_0^t \bigl(\mu_{\ep(s)}+\sigma_{\ep(s)}\psi_{\ep(s)}\bigr)\rmd s,\quad Z_{t}=\int_{0}^{t}\sigma_{\ep(s)}\rmd W_{s}^{\qq} \quad\text{and}\quad J_t=\int_0^t \eta_{\ep(s^-)}\rmd N_s=\sum_{j=1}^{N_t}\eta_{\ep_j}.
\end{equation*}
Following similar steps as in the non-diffusive case, we obtain
\begin{align}
\E^{\qq}\left[r_{T}\mid\mathcal{F}_{t}\right] &=r_{t}+\E_{i}^{\qq}\Bigl[\widetilde{X}_{T-t}+\widetilde{Z}_{T-t}+\widetilde{J}_{T-t}\Bigr] \notag\\
&=r_{t}+\E_{i}^{\qq}\Bigl[\widetilde{X}_{T-t}+\widetilde{J}_{T-t}\Bigr]=r_{t}+\E_{i}^{\qq}\Bigl[\widetilde{Y}_{T-t}\Bigr],\label{ejem 3}
\end{align} 
because $\E_{i}^{\qq}[\widetilde{Z}_{T-t}]=0$ for $i=0,1$. The last expected value exhibits the same functional formula as \eqref{eq:mYtilde}, replacing $d_i$ by $\tilde{d}_i=\mu_i+\sigma_{i}\psi_{i}+\lambda_i^{\qq}\eta_i$, $i=0,1$.

\subsection{Jump-telegraph diffusion Dothan model}
By adding a diffusion term to the jump-telegraph Dothan Model, we obtain the following stochastic evolution for the short rate:
\begin{equation*}
\rmd r_t=r_{t^-}\bigl(\mu_{\ep(t)}\rmd t+\sigma_{\ep(t)}\rmd W_t+\eta_{\ep(t^{-})}\rmd N_t\bigr)
\end{equation*}
which, under measure $\q$,  is written as follows:
\begin{equation*}
\rmd r_t=r_{t^-}\bigl((\mu_{\ep(t)}+\sigma_{\ep(t)}\psi_{\ep(t)})\rmd t+\sigma_{\ep(t)}\rmd W_t^{\qq}+\eta_{\ep(t^{-})}\rmd N_t\bigr)
\end{equation*}
The solution is written in terms of the stochastic exponential as follows:
\begin{align*}
r_{t} & =r_{0}\mathcal{E}_{t}\left(X+Z+J\right)\\
&=r_{0}\exp\left(\int_{0}^{t}(\mu_{\ep(s)}+\sigma_{\ep(s)}\psi_{\ep(s)})\rmd s+\int_0^t\sigma_{\ep(s)} \rmd W_s-\frac12\int_0^t\sigma_{\ep(s)}^2\rmd s\right)\prod_{j=1}^{N_{t}}\left(1+\eta_{\varepsilon_{j}}\right)\\
&=r_{0}\exp\left(\int_{0}^{t}\left(\mu_{\ep(s)}+\sigma_{\ep(s)}\psi_{\ep(s)}-\frac12\sigma_{\ep(s)}\right)\rmd s+\int_0^t \sigma_{\ep(s)}\rmd W_s+\int_{0}^{t}\log(1+\eta_{\ep(s^{-})})\rmd N_{s}\right)\\
  & = r_0\exp\left(\widehat{X}_t+Z_t+\widehat{J}_t\right).
\end{align*}
To establish a useful formula for the expected value of the future spot rate, we require that $\sigma_0=\sigma_1=\sigma$. Following similar steps as those performed for the non-diffusive case, we obtain
\begin{align}
\E^{\qq}\left[r_{T}\mid\mathcal{F}_{t}\right] & =r_{t}\E_i^{\qq}\left[\exp\left(\widetilde{X}_{T-t}+\widetilde{Z}_{T-t}+\widetilde{J}_{T-t}\right)\right] \notag \\
&=r_t\E_{i}^{\qq}\left[\exp\left(\widetilde{Y}_{T-t}\right)\right]\E_i^{\qq}\left[\exp\left(\widetilde{Z}_{T-t}\right)\right] \notag \\
&=r_t\E_{i}^{\qq}\left[\exp\left(\widetilde{Y}_{T-t}\right)\right]\exp\left(\frac{\sigma^2}{2}(T-t)\right).\label{espmodel4}
\end{align}
Again, in this case, the expected value exhibits the same functional form as that in \eqref{eq:ME}, replacing $\mu_i$ by $\tilde{\mu}_i=\mu_i+\sigma\psi-\sigma^2/2$ for $i=0,1$.

\section{Zero coupon bond price and numerical results}\label{sec5}
We will calculate the price of a zero coupon bond for each model of the previous section by two methods: first, using standard numerical methods for the solution of the Cauchy problem in \eqref{eq:DFi}; second, using  the unbiased expectation hypothesis, together with the relationship
\begin{equation}\label{eq:aproxB}
F_{\ep(t)}(t,r_t)=\exp\left(-\int_t^T f_{\ep(t)}(t,s,r_t)\rmd s\right).
\end{equation}

\subsection{Zero coupon bond price in the jump-telegraph Merton model}
In the jump-telegraph Merton model the system \eqref{eq:DFi} reduces to
\begin{equation*}
\begin{cases}
\dfrac{\partial F_{i}}{\partial t}(t,x)+\mu_{i}\dfrac{\partial F_{i}}{\partial x}(t,x) +\lambda_{i}^{\qq}\left[F_{1-i}(t,x+\eta_{i})-F_{i}(t,x)\right]=xF_{i}\left(t,x\right),\quad i=0,1\\[2mm]
F_{0}(T,x)=F_{1}(T,x) =1.
\end{cases}
\end{equation*}
The equations can be solved by suitable variable separation and reduced to a coupled system of ordinary differential equations. In Table \ref{tabla1}, some numerical results are shown for specific parameters.

Meanwhile, by \eqref{eq:aproxB} and \eqref{eq:HEM1}, the approximated price of the zero coupon bond obtained from the expectation hypotheses is given by
\begin{align}
\exp\left(-\int_{t}^{T}\E^{\qq}[r_{s}\mid\mathcal{F}_{t}] \rmd s\right) &=\exp\left(-\int_{t}^{T}r_{t}\rmd s-\int_{t}^{T}\E_{i}^{\qq}[\widetilde{Y}_{s-t}]\rmd s\right) \notag \\
&=\exp\bigl(r_{t}C(t,T)+D_{i}(t,T)\bigr). \label{bonomodel1}
\end{align}
Here,
\[
C(t,T)=-(T-t),
\]
and
\begin{align*}
D_{i}&(t,T)=\\
&-\frac{1}{2\lambda^{\qq}}\left[\left(\lambda_{1}^{\qq}d_{0}+\lambda_{0}^{\qq}d_{1}\right)\frac{\left(T-t\right)^{2}}{2}+\left(-1\right)^{i}\lambda_{i}^{\qq}(d_{0}-d_{1})\left(\frac{2\lambda^{\qq}\left(T-t\right)+\rme^{-2\lambda^{\qq}\left(T-t\right)}-1}{2\lambda^{\qq}\cdot2\lambda^{\qq}}\right)\right].
\end{align*}
Observe that, under the unbiased expectations hypothesis, the bond price exhibits an affine structure.

\textbf{Numerical results:} Data provided for the following parameter values: $\mu_0=-0.02,\mu_1=0.05,\lambda_0^{\qq}=1,\lambda_1^{\qq}=2,\eta_0=0.01,\eta_1=-0.02$.

\begin{table}[ht!]
\begin{centering}
\begin{tabular}{|c|c|c|c|c|}
\hline 
{\textbf{Initial rate: $5\%$}}& \multicolumn{2}{c}{\textbf{ODE}} & \multicolumn{2}{|c|}{\textbf{Expectation}}\\
\hline
\hline 
Maturity & $F_{0}$ & $F_{1}$& $F_{0}$ & $F_{1}$\tabularnewline
\hline  
1 month & $0.995875$ & $0.995811$ & 0.995875 & 0.995811    \\
\hline 
1 quarter & $0.987844$ & $0.987358$ & 0.987843 & 0.987355     \\
\hline 
1 semester & $0.976244$ & $0.974689$ & 0.976239 & 0.974672 \\
\hline 
1 year & $0.954317$ & $0.950064$ & 0.954264 & 0.949927  \\
\hline 
\end{tabular}
\par\end{centering}
\caption{\textbf{Zero coupon bond prices: Jump-telegraph Merton model.}\label{tabla1}
}
\end{table}

\subsection{Zero coupon bond price in the jump-telegraph Dothan model}\label{sub:52}
In the jump-telegraph Dothan case, the system \eqref{eq:DFi} reduces to
\begin{equation*}
\begin{cases}
\dfrac{\partial F_{i}}{\partial t}(t,x)+x\mu_{i}\dfrac{\partial F_{i}}{\partial x}(t,x)+\lambda_{i}^{\qq}\left[F_{1-i}(t,x(1+\eta_{i}))-F_{i}(t,x)\right]=xF_{i}(t,x),\quad i=0,1,\\
F_{0}(T,x)=F_{1}(T,x)=1.
\end{cases}
\end{equation*}
The system is solved numerically using a finite difference up-wind scheme for the transport terms.

By \eqref{eq:aproxB}, \eqref{eq:HEM2}, and \eqref{eq:ME}, the price of the bond, using the approximation provided by the expectation hypothesis, is given by
\begin{align*}
\exp\left(-\int_{t}^{T}\E\left[r_{s}\mid\mathcal{F}_{t}\right]ds\right)  &=\exp\left(-r_t\int_{t}^{T}\E_{i}^{\qq}\Bigl[\exp\bigl(\widetilde{Y}_{s-t}\bigr)\Bigr]\rmd s\right)\\
&=\exp\left(-r_t\bigl[G(t,T)+E_i H(t,T)\bigr]\right),
\end{align*}
where
\begin{align*}
G(t,T)&=\frac{\rme^{(T-t)(\zeta-\lambda^{\qq})}\left[(\zeta-\lambda^{\qq})\cosh\bigl((T-t)\sqrt{D}\bigr)-\sqrt{D}\sinh\bigl((T-t)\sqrt{D}\bigr)\right]-(\zeta-\lambda^{\qq})}{(\zeta-\lambda^{\qq})^2-D},\\[1mm]
E_i&=(-1)^i\frac{\chi-\nu+(-1)^i\lambda_{i}^{\qq}\left(1+\eta_{i}\right)}{\sqrt{D}},\quad i=0,1,\quad \text{and} \\[1mm]
H(t,T)&=\frac{\rme^{(T-t)(\zeta-\lambda^{\qq})}\left[(\zeta-\lambda^{\qq})\sinh\bigl((T-t)\sqrt{D}\bigr)-\sqrt{D}\cosh\bigl((T-t)\sqrt{D}\bigr)\right]+\sqrt{D}}{(\zeta-\lambda^{\qq})^2-D}.
\end{align*}
\textbf{Numerical results:} The solution is provided for the following parameter values: $\mu_0=-0.1,\mu_1=0.25, \lambda_0^{\qq}=1.0, \lambda_1^{\qq}=2.0,\eta_0=0.1,\eta_1=-0.2$.
\begin{table}[ht!]
\begin{centering}
\begin{tabular}{|c|c|c|c|c|}
\hline 
{\textbf{Initial rate: $5\%$}}& \multicolumn{2}{c}{\textbf{Finite Differences}} & \multicolumn{2}{|c|}{\textbf{Expectation}}\\
\hline
\hline 
Maturity & $F_{0}$ & $F_{1}$& $F_{0}$ & $F_{1}$ \tabularnewline
\hline  
1 month & $0.995842$ & $0.995869$ & $0.995843$ & $0.995867$ \\
\hline 
1 quarter & $0.987594$ & $0.987786$& $0.987596$  & $0.987781$ \\
\hline 
1 semester & $0.975430$ & $0.976039$& $0.975431$ & $0.976029$\\
\hline 
1 year & $0.951962$ & $0.953645$& $0.951955$ & $0.953615$ \\
\hline 
\end{tabular}
\par\end{centering}
\caption{\textbf{Zero coupon bond prices: Jump-telegraph Dothan model.}\label{tabla2}
}
\end{table}

\subsection{Zero coupon bond price in the jump-telegraph diffusion Merton model.}
In the case of a jump-telegraph Merton model with diffusion, the system \eqref{eq:DFi} reduces to
\begin{equation}
\begin{cases}
\dfrac{\partial F_{i}}{\partial t}(t,x)+(\mu_{i}+\psi_{i}\sigma_{i})\dfrac{\partial F_{i}}{\partial x}(t,x)+\dfrac{\sigma_{i}^{2}}{2}\dfrac{\partial^{2}F_{i}}{\partial x^{2}}(t,x) +\lambda_{i}^{\qq}\left[F_{1-i}(t,x+\eta_{i})-F_{i}(t,x)\right] =xF_{i}\left(t,x\right),\\[1mm]
i=0,1,\\
F_{0}\left(T,x\right)  = F_{1}\left(T,x\right)=1.
\label{eq:75-1-1-1}
\end{cases}
\end{equation}
This coupled system can be transformed in a system of two ordinary differential equations. The effect of $\lambda^{\qq}$ in the first term is negligible for long maturities. The opposite occurs with volatility.

To obtain the zero coupon bond prices under the expectation hypothesis, we applied the fact that the expected value  $\E\left[Z_{t}\right]$ is equal to zero. The calculations are similar to those performed in the jump-telegraph Merton model, and the bond price matches that in \eqref{bonomodel1} but with $\tilde{d}_i=\mu_i+\sigma_i\psi_i+\eta_i\lambda_i^{\qq}$, in place of $d_i$, $i=0,1$. 

\textbf{Numerical results}: Data provided for the following parameter values: $\mu_0=-0.02,\mu_1=0.05,\psi_0=0.5,\psi_1=1.0,\sigma_0=0.02,\sigma_1=0.06,\lambda_0^{\qq}=1.0,\lambda_1^{\qq}=2.0,\eta_0=0.01,\eta_1=-0.02$.
\begin{table}[ht!]
\begin{centering}
\begin{tabular}{|c|c|c|c|c|c|c|}
\hline 
{\textbf{Initial rate: $5\%$}}& \multicolumn{2}{c}{\textbf{ODE}} & \multicolumn{2}{|c|}{\textbf{Expectation}} \\
\hline
\hline 
Maturity & $F_{0}$ & $F_{1}$& $F_{0}$ & $F_{1}$ \tabularnewline
\hline  
1 month & 0.995836 & 0.995613 & 0.995836  &	0.995613\\
\hline 
1 quarter & 0.987429 & 0.985732 &	0.987427&	0.985721
  \\
\hline 
1 semester & 0.974318 & 0.968920 &  0.974294 &	0.968830 \\
\hline 
1 year & 0.945471 & 0.930939 &	0.945206 & 0.930256 \\
\hline 
\end{tabular}
\par\end{centering}
\caption{\textbf{Zero coupon bond prices: Jump-telegraph diffusion Merton model.}\label{tabla3}
}
\end{table}


\subsection{Zero coupon bond price in the jump-telegraph diffusion Dothan model}

Under this model, the system \eqref{eq:DFi} for the prices of the zero-coupon bonds is given by
\begin{equation}
\begin{cases}
\dfrac{\partial F_{i}}{\partial t}(t,x)+(\mu_{i}+\psi_{i}\sigma_{i})x\dfrac{\partial F_{i}}{\partial x}(t,x)+\dfrac{\sigma_{i}^{2}x^2}{2}\dfrac{\partial^{2}F_{i}}{\partial x^{2}}(t,x) +\lambda_{i}^{\qq} \left[F_{1-i}(t,x(1+\eta_{i}))-F_{i}(t,x)\right]\\
=xF_{i}\left(t,x\right),\quad i=0,1,\\[1mm]
F_{0}\left(T,x\right) =F_{1}\left(T,x\right)  =1.
\end{cases}\label{eq:DyS}
\end{equation}
The system is solved numerically by an implicit, up-wind finite difference scheme. 
Moreover, from \eqref{eq:aproxB} and \eqref{espmodel4}, we obtain that the price of the bond under the expectation hypothesis can be approximated by
\begin{align*}
\exp\left(-\int_{t}^{T}\E\left[r_{s}\mid\mathcal{F}_{t}\right]ds\right)&=\exp\left(-r_t\int_{t}^{T}\E_{i}^{\qq}\left[\exp\left(\widetilde{Y}_{s-t}\right)\right]\exp\left(\frac{\sigma^2}{2}(s-t)\right)\rmd s\right) \\
&=\exp\left(-r_t\bigl[\widetilde{G}(t,T)+E_i \widetilde{H}(t,T)\bigr]\right), 
\end{align*}
where
\begin{align*}
\widetilde{G}(t,T)&=\frac{\rme^{(T-t)(\zeta-\lambda^{\qq}+\sigma^2/2)}\left[(\zeta-\lambda^{\qq}+\sigma^2/2)\cosh\bigl((T-t)\sqrt{D}\bigr)-\sqrt{D}\sinh\bigl((T-t)\sqrt{D}\bigr)\right]-(\zeta-\lambda^{\qq}+\sigma^2/2)}{(\zeta-\lambda^{\qq}+\sigma^2/2)^2-D},\\[1mm]
\widetilde{H}(t,T)&=\frac{\rme^{(T-t)(\zeta-\lambda^{\qq}+\sigma^2/2)}\left[(\zeta-\lambda^{\qq}+\sigma^2/2)\sinh\bigl((T-t)\sqrt{D}\bigr)-\sqrt{D}\cosh\bigl((T-t)\sqrt{D}\bigr)\right]+\sqrt{D}}{(\zeta-\lambda^{\qq}+\sigma^2/2)^2-D}.
\end{align*}
Here $E$ is defined in Section \ref{sub:52}, and $\zeta$, $\chi$, $\nu$, $D$ are given in \eqref{eq:ME} with $\tilde{\mu}_i=\mu_i+\sigma\psi-\sigma^2/2$ in place of $\mu_i$, $i=0,1$.

\textbf{Numerical results}: The solution is provided for the following parameter values: $\mu_0=-0.1,\mu_1=0.25,\psi_0=\psi_1=1.0,\sigma_0=\sigma_1=0.4,\lambda_0^{\qq}=1.0, \lambda_1^{\qq}=2.0,\eta_0=0.1,\eta_1=-0.2$. 
\begin{table}[ht!]
\begin{centering}
\begin{tabular}{|c|c|c|c|c|}
\hline 
{\textbf{Initial rate: $5\%$}}& \multicolumn{2}{c}{\textbf{Finite differences}} & \multicolumn{2}{|c|}{\textbf{Expectation}}\\
\hline
\hline 
Maturity & $F_{0}$ & $F_{1}$& $F_{0}$ & $F_{1}$  \tabularnewline
\hline  
1 month & $0.995774$ & $0.995798$ & $0.995773$ & $0.995797$  \\
\hline 
1 quarter & $0.986965$ & $0.987161$& $0.986959$ & $0.987156$  \\
\hline 
1 semester & $0.972865$ & $0.973544$& $0.972844$ & $0.973522$ \\
\hline 
1 year & $0.941475$ & $0.943588$& $0.941334$  & $0.943434$ \\
\hline 
\end{tabular}
\par\end{centering}
\caption{\textbf{Zero coupon bond prices: Jump-telegraph diffusion Dothan model.}\label{tabla 4}
}
\end{table}

\section{Conclusions}
We herein presented new results in two main directions. First, we used the properties of the telegraph processes with jumps to prove a Feynman--Ka\v{c} representation theorem for zero coupon bond pricing. Meanwhile, an analytical solution was obtained for different dynamics of the short rate under telegraph processes with jumps. To obtain this solution, the unbiased expectation hypothesis was assumed. Furthermore, it was shown that the cost of assuming this hypothesis was low, because the numerical results for the Cauchy problem were sufficiently close to the analytical approximations.

To our best knowledge, none have incorporated these types of processes for fixed income modeling, nor analytical formulas to approximate the prices of zero coupon bonds when the dynamics of the short rate allow for jumps and regime switches. Owing to the significant interest in the effect of these aspects on the yield curve, we believe that this work will be beneficial.  

A natural extension of this work is the search of analytical formulas for more robust models that include, for instance, mean reversion (CIR, Hull--White) and the valuation of more complex fixed income instruments.



\end{document}